\documentclass[12pt]{article}

\usepackage{geometry}
\usepackage{cite}
\usepackage{paralist}
\usepackage{graphicx}
\usepackage{subfigure}
\usepackage{amssymb}
\usepackage{amsmath}
\allowdisplaybreaks[1]

\usepackage{amsthm}

\newtheorem{theorem}{Theorem}
\newtheorem{lemma}[theorem]{Lemma}

\newtheorem{observation}[theorem]{Observation}

\makeatletter
\def\@endtheorem{\endtrivlist}
\makeatother

\newcounter{rrule}
\newenvironment{rrule}{\refstepcounter{rrule}\par\smallskip\noindent
\textbf{(R\arabic{rrule})}\quad}{}
\newcommand{\currentrule}{R\arabic{rrule}}

\begin{document}

\title{Kernel for $K_t$-free edge deletion}
\author{Dekel Tsur%
\thanks{Ben-Gurion University of the Negev.
Email: \texttt{dekelts@cs.bgu.ac.il}}}
\date{}
\maketitle

\begin{abstract}
In the \emph{$K_t$-free edge deletion} problem, the input is a graph $G$ and an
integer $k$, and the goal is to decide whether there is a set of at most $k$
edges of $G$ whose removal results a graph with no clique of size $t$.
In this paper we give a kernel to this problem with $O(k^{t-1})$ vertices and
edges.
\end{abstract}


\section{Introduction}

In the \emph{$H$-free edge deletion} problem, the input is a graph $G$ and an
integer $k$, and the goal is to decide whether there is a set of at most $k$
edges of $G$ whose removal results a graph that does not contain $H$ as an
induced graph.
The kernelization of $H$-free edge deletion problems have been studied in
several papers.
Cai and Cai~\cite{cai2015incompressibility} showed that unless
$\mathrm{NP} \subseteq \mathrm{coNP/poly}$, $H$-free edge deletion does not
have a polynomial kernel in the following cases:
(1) $H$ is 3-connected and is not a complete graph
(2) $H$ is a path or a cycle with at least 4 edges.
On the positive side, polynomial kernels were given for the cases when $H$
is a $P_3$~\cite{gramm2005graph},
a $P_4$~\cite{guillemot2013non},
a diamond~\cite{fellows2011graph,sandeep2015parameterized,cao2018polynomial},
and a $K_t$ (a clique with $t$ vertices)~\cite{cai2012polynomial}.

For $K_t$-free edge deletion, the kernel of Cai~\cite{cai2012polynomial} has
$O(k^{t(t-1)/2})$ vertices and edges.
In this paper we give a kernel for $K_t$-free edge deletion with
$O(k^{t-1})$ vertices and edges.

\paragraph{Preliminaries}
For a graph $G$, a set of edges whose removal results a graph with no
clique of size $t$ is called a \emph{deletion set} of $G$.

In the \emph{hitting set} problem the input is a set
$\mathcal{F}$ of subsets of a set $U$ and an integer $k$,
and the goal is to decide whether there is a set $S \subseteq U$ of size at
most $k$ such that $S \cap F \neq \emptyset$ for every $F \in \mathcal{F}$.
Such a set $S$ is called a \emph{hitting set} of $\mathcal{F}$.

Let $\mathcal{F}$ be a set of subsets of a set $U$, and let
$S_1,\ldots,S_p \in \mathcal{F}$ be distinct sets.
We say that $S_1,\ldots,S_k$ is a \emph{sunflower} if there is a set $Y$
such that $S_i \cap S_j = Y$ for every $i\neq j$,
and $S_i \setminus Y \neq \emptyset$ for every $i$.
The set $Y$ is called the \emph{core} of the sunflower.

\begin{lemma}[Sunflower lemma~\cite{erdos1960intersection}]\label{lem:sunflower}
Let $\mathcal{F}$ be a set of subsets of a set $U$, and every set in
$\mathcal{F}$ has size at most $d$.
If $|\mathcal{F}| > 2\cdot d!(k-1)^d$ then $\mathcal{F}$ contains a sunflower
$S_1,\ldots,S_k$.
Moreover, such a sunflower can be found in polynomial time.
\end{lemma}

\section{The kernel}
Given an instance $(G,k)$ of $K_t$-free edge deletion,
the kernelization algorithm consists of three stages.
In the first stage, the algorithm creates an instance $(\mathcal{F},k)$ of
hitting set as follows.
For a set of vertices $X$, let $E_X = \{(x,y)\colon x,y\in X, x\neq y\}$.
The set $\mathcal{F}$ contains a set $E_X$ for every clique $X$ of size $t$
in $G$.
\begin{observation}\label{obs:S}
A set of edges $S$ is a deletion set of $G$ if and only if
$S$ is a hitting set of $\mathcal{F}$.
\end{observation}
In the second stage that kernelization algorithm creates an instance
$(\mathcal{F}',k')$ of hitting set
that is equivalent to $(\mathcal{F},k)$ and $|\mathcal{F}'| = O(k^{t-1})$.
In the third stage, the algorithm creates an instance $(G',k')$ of
$K_t$-free edge deletion that is equivalent to $(G,k)$.
The number of vertices in $G'$ is at most
$t \cdot |\mathcal{F}'| = O(k^{t-1})$
and the number of edges is at most
$\binom{t}{2} \cdot  |\mathcal{F}'| = O(k^{t-1})$.

We now describe the second stage of the kernelization algorithm.
For every edge $(x,y) \in E(G)$,
let $\mathcal{F}_{xy} = \{X \subseteq V(G) \setminus\{x,y\}\colon
E_{X \cup \{x,y\}} \in \mathcal{F}\}$.
The kernelization algorithm uses the following reduction rule.
\begin{rrule}
If there is an edge $(x,y) \in E(G)$ such that
$|\mathcal{F}_{xy}| > 2\cdot(t-2)!\cdot k^{t-2}$,
find a sunflower $X_1,\ldots,X_{k+1}$ in $\mathcal{F}_{xy}$
and let $Y$ be the core of the sunflower.
For every $E_Z \in \mathcal{F}$ such that $Y \cup \{x,y\} \subseteq Z$,
remove $E_Z$ from $\mathcal{F}$.
Additionally, add the set $E_{Y \cup \{x,y\}}$ to
$\mathcal{F}$.\label{rule:sunflower}
\end{rrule}

Denote by $\mathcal{F}'$ the set $\mathcal{F}$ after the application of
Rule~(\currentrule).
\begin{lemma}\label{lem:S}
Let $S$ be a set of edges of $G$ of size at most $k$.
$S$ is a hitting set of $\mathcal{F}$ if and only if $S$ is a hitting set of
$\mathcal{F}'$.
\end{lemma}
\begin{proof}
Suppose that $S$ is a hitting set of $\mathcal{F}$.
Since $X_i \setminus Y \neq X_j \setminus Y$ for every $i \neq j$,
there is an index $i$ such that $S$ does not contain an edge with at least
one endpoint in $X_i \setminus Y$.
Since $S \cap E_{X_i \cup \{x,y\}} \neq \emptyset$,
$S$ contains an edge from $E_{Y \cup \{x,y\}}$.
Note that $E_{Y \cup \{x,y\}}$ is the only set in
$\mathcal{F}' \setminus \mathcal{F}$.
Therefore, $S$ is a hitting set of $\mathcal{F}'$.

Now suppose that $S$ is a hitting set of $\mathcal{F}'$.
By definition, if  $E_{Z} \in \mathcal{F} \setminus \mathcal{F}'$ then
$Y \cup \{x,y\} \subseteq Z$.
Since $S' \cap E_{Y \cup\{x,y\}} \neq \emptyset$, we have that
$S' \cap E_Z \neq \emptyset$.
Therefore, $S'$ is a hitting set of $\mathcal{F}$.
\end{proof}

The kernelization algorithm applies Rule~(R\ref{rule:sunflower}) until the
rule is not applicable,
and denote by $(\mathcal{F}',k)$ the resulting instance.
The algorithm then applies the following rule.

\begin{rrule}
If $|\mathcal{F}'|>2\cdot(t-2)!\cdot k^{t-1}$
return a fixed no instance.
\end{rrule}
\begin{lemma}
Rule~(\currentrule) is safe.
\end{lemma}
\begin{proof}
Suppose that $(\mathcal{F}',k)$ is a yes instance, and let $S$ be a hitting
set of $\mathcal{F}'$ of size at most $k$.
It suffices to show that for every $(x,y) \in S$,
the number of sets in $\mathcal{F}'$ that contain $(x,y)$ is at most
$2 \cdot (t-2)! \cdot k^{t-2}$.
Fix some $(x,y) \in S$.
A set in $\mathcal{F}'$ that contains $(x,y)$ is of the form
$E_{X \cup \{x,y\}}$ and, by definition, $X \in \mathcal{F}_{xy}$.
Therefore, the number of sets in $\mathcal{F}'$ that contain $(x,y)$ is
at most $|\mathcal{F}_{xy}|$.
Since Rule~(R\ref{rule:sunflower}) cannot be applied,
$|\mathcal{F}_{xy}| \leq 2 \cdot (t-2)! \cdot k^{t-2}$.
\end{proof}

We now describe the third stage of the kernelization algorithm.
The algorithm builds a graph $G' = (V',E')$ as follows.
It first initializes $V' = \emptyset$ and $E' = \emptyset$.
Then, for every $E_X \in \mathcal{F}$,
if $|X| = t$, the algorithm performs 
$V' \gets V' \cup X$ and $E' \gets E' \cup E_X$.
If $|X| < t$, the algorithm performs
$V' \gets V' \cup X \cup V_X$, where $V_X$ is a set containing $t-|X|$ new
vertices, and $E' \gets E' \cup E_{X\cup V_X}$.

\begin{lemma}
$(G,k)$ is a yes instance if and only if $(G',k)$ is a yes instance.
\end{lemma}
\begin{proof}
Suppose that $(G,k)$ is a yes instance and let $S$ be a deletion set of $G$
of size at most $k$.
Let $Y$ be a clique in $G'$ of size $t$.
If $Y \subseteq V(G)$ then $Y$ is also a clique in $G$, and therefore $S$
contains at least one edge from $E_Y$.
Otherwise, $Y$ contains a vertex from a set $V_X$ for some
$E_X \in \mathcal{F}'$.
By construction, a vertex in $V_X$ is not adjacent to vertices in
$V(G') \setminus (X \cup V_X)$.
Therefore, $Y = X \cup V_X$.
By Observation~\ref{obs:S} and Lemma~\ref{lem:S}, $S$ is a hitting set of
$\mathcal{F}'$. Since $E_X \in \mathcal{F}'$, we have that
$S$ contains an edge from $E_X$, and therefore $S$ contains an edge from $E_Y$.
It follows that $S$ is a deletion set of $G'$.

For the opposite direction,
suppose that $(G',k)$ is a yes instance and let $S$ be a deletion set of $G'$
of size at most $k$.
We claim that $(\mathcal{F}',k)$ is a yes instance.
Let $E_X \in \mathcal{F}'$.
If $|X| = t$ then $X$ is a clique in $G'$ and therefore $S$ contains an edge
from $E_X$.
If $|X| < t$ and $S$ does not contain an edge from $E_X$ then
$S$ must contain an edge $e$ from $E_{X \cup V_X} \setminus E_X$.
Let $e'$ be some edge from $E_X$.
Then, the set $S' = (S \setminus \{e\}) \cup \{e'\}$ is also a deletion set
of $G'$ of size $k$, and $S'$ contains an edge from $E_X$.
Therefore, there is a deletion set of $G'$ of size at most $k$ that
contains an edge from $E_X$ for every $E_X \in \mathcal{F}'$.
Thus, $(\mathcal{F}',k)$ is a yes instance.
By Observation~\ref{obs:S} and Lemma~\ref{lem:S}, $(G,k)$ is a yes instance.
\end{proof}

From the lemmas above, we obtain the result of this paper.
\begin{theorem}
$K_t$-free edge deletion has a kernel with $O(k^{t-1})$ vertices and edges.
\end{theorem}

\bibliographystyle{plain}
\bibliography{k4}

\end{document}